\documentclass[11pt,letterpaper]{article}
\usepackage[margin=1in]{geometry}
\usepackage[T1]{fontenc}
\usepackage[utf8]{inputenc}
\usepackage[english]{babel}
\usepackage{tikz,amsmath,amssymb,amsfonts,latexsym,graphicx,amsthm,algorithmicx,comment}
\usepackage{bm}
\usepackage{caption}
\usepackage{subcaption}
\usepackage{afterpage}
\usepackage{todonotes}
\usepackage{multirow}

\usepackage{thmtools}
\declaretheorem{theorem}

\usepackage{algorithm}
\usepackage[noend]{algpseudocode}
\usepackage{hyperref}
\usepackage[capitalise]{cleveref}

\makeatletter
\newcounter{phase}[algorithm]
\newlength{\phaserulewidth}

\makeatother

\theoremstyle{plain}
\newtheorem{lemma}[theorem]{Lemma}

\newtheorem{fact}[theorem]{Fact}

\theoremstyle{definition}

\newtheorem{definition}{Definition}

\newcommand{\OPT}{\mathrm{OPT}}
\newcommand{\TSP}{\mathrm{TSP}}

\newcommand{\opt}{\mathrm{opt}}
\newcommand{\bd}{\mathrm{bd}}
\newcommand{\dist}{\mathrm{dist}}
\newcommand{\cost}{\mathrm{cost}}
\newcommand{\thr}{\Gamma}

\newcommand{\eps}{\epsilon}
\newcommand{\demand}{\mathrm{demand}}

\title{Unsplittable Euclidean Capacitated Vehicle Routing: \\ A $(2+\eps)$-Approximation Algorithm}

\author{
Fabrizio Grandoni\footnote{IDSIA, USI-SUPSI, Switzerland. E-mail: \texttt{fabrizio@idsia.ch}. Partially supported by the SNSF Grant 200021 200731 / 1.} 
\and
Claire Mathieu\footnote{CNRS, Paris, France. E-mail: \texttt{claire.mathieu@irif.fr}. Partially supported by the grant ANR-19-CE48-0016 from the French National Research Agency (ANR).}  
\and 
Hang Zhou\footnote{Ecole Polytechnique, Institut Polytechnique de Paris, France. E-mail: \texttt{hzhou@lix.polytechnique.fr}. Partially supported by the grant on ``Efficiency in Algorithms'' from Hi! PARIS.}
}

\date{}
\begin{document}

\maketitle

\begin{abstract}
In the unsplittable capacitated vehicle routing problem, we are given a metric space with a vertex called depot and a set of vertices called terminals. Each terminal is associated with a positive demand between 0 and 1. The goal is to find a minimum length collection of tours starting and ending at the depot such that the demand of each terminal is covered by a single tour (i.e., the demand cannot be split), and the total demand of the terminals in each tour does not exceed the capacity of 1.

Our main result is a polynomial-time $(2+\epsilon)$-approximation algorithm for this problem in the two-dimensional Euclidean plane, i.e., for the special case where the terminals and the depot are associated with points in the Euclidean plane and their distances are defined accordingly. This improves on recent work by Blauth, Traub, and Vygen [IPCO'21] and Friggstad, Mousavi, Rahgoshay, and Salavatipour [IPCO'22].
\end{abstract}

\thispagestyle{empty}
\setcounter{page}{0}

\newpage
\pagenumbering{arabic}

\setcounter{page}{1}

\section{Introduction}
In the \emph{unsplittable capacitated vehicle routing problem (unsplittable CVRP)}, we are given a metric space with a vertex $O$ called \emph{depot} and a set of $n$ vertices $V$ called \emph{terminals}. Each terminal $v\in V$ is associated with a positive demand $demand(v)\in (0,1]$. 
A feasible solution $S$ is a collection of tours starting and ending at the depot $O$ such that the demand of each terminal $v\in V$ is covered by a \emph{single} tour in $S$ (i.e.,\ the demand of $v$ cannot be split), and, for each tour $t\in S$, the total demand $demand(t)$ of the terminals covered by $t$ does not exceed the \emph{tour capacity} of 1.
Our goal is to find a feasible solution $S$ minimizing the total length of the tours $cost(S):=\sum_{t\in S}cost(t)$, where $cost(t)$ denotes the length of $t$, i.e., the overall weight of the edges of $t$. As an application, the reader might think about a set of identical vehicles, each with the same capacity, located in a depot; These vehicles have to deliver a set of items at different locations and return to the depot.

Originally introduced by Dantzig and Ramser in 1959~\cite{dantzig1959truck}, the unsplittable CVRP is arguably one of the most basic problems in Operations Research, and it generalizes famous problems in a natural way. For example, if the sum of the demands is at most $1$ (hence one tour is sufficient), the problem is equivalent to Traveling Salesman Problem (TSP); if all the terminals are placed at the same location, the problem is equivalent to Bin Packing.
Since Bin Packing is (3/2)-hard to approximate (i.e., a better than $3/2$ polynomial-time approximation algorithm does not exist unless $P=NP$)~\cite{williamson2011design}, the unsplittable CVRP is also  $(3/2)$-hard to approximate.

On general metrics, the first constant-factor approximation algorithm (namely a $3.5$ approximation) for the unsplittable CVRP is the \emph{iterated tour partitioning (ITP)} algorithm, which was proposed and analyzed in the 1980s by
Haimovich and Rinnooy~Kan~\cite{haimovich1985bounds} and Altinkemer and  Gavish~\cite{altinkemer1987heuristics}. The approximation ratio for the unsplittable CVRP was only recently slightly improved by Blauth, Traub, and Vygen~\cite{blauth2021improving}, and then further improved to roughly $3.194$ by Friggstad, Mousavi, Rahgoshay, and Salavatipour~\cite{friggstad2021improved}.

We study the unsplittable CVRP in the \emph{two-dimensional Euclidean plane}, called the \emph{unsplittable Euclidean CVRP}. Here the depot and terminals correspond to points in the Euclidean plane, and for any two vertices $u,v\in V\cup \{O\}$, the weight of the edges $\{u,v\}$ is given by the corresponding Euclidean distance. We remark that the $(3/2)$-hardness of approximation mentioned before extends to this case (indeed, even to the one-dimensional Euclidean case).
The approach in \cite{friggstad2021improved} directly implies a $2.694$ approximation for the unsplittable Euclidean CVRP, which is also the best known approximation factor for this problem. More precisely, the algorithm in \cite{friggstad2021improved} has approximation ratio strictly below $1.694+\alpha$, where $\alpha$ is the approximation ratio of a TSP algorithm in the considered metric. Since the Euclidean TSP admits a PTAS~\cite{arora1998polynomial,mitchell1999guillotine}, the algorithm in~\cite{friggstad2021improved} leads to a $2.694$ approximation.

Our main result is a $(2+\eps)$-approximation algorithm for the unsplittable Euclidean CVRP (\cref{thm:main}). This matches the best known approximation factor for the special case when all the demands are equal \cite{altinkemer1990heuristics}.

\begin{theorem}
\label{thm:main}
Let $\eps>0$.
There is a polynomial time $(2+\eps)$-approximation algorithm for the unsplittable Euclidean capacitated vehicle routing problem.
\end{theorem}

As a by-product, we design a \emph{polynomial time approximation scheme (PTAS)} for the case when all demands are \emph{big} (\cref{thm:big-demands}). This matches the best known approximation factor for the special case when all the demands are big and equal \cite{haimovich1985bounds}.
\begin{definition}
Let $\eps>0$. We say that a terminal is \emph{big} (w.r.t.\ $\eps$) if its demand is at least $\eps$, and it is \emph{small} otherwise.
\end{definition}
\begin{theorem}
\label{thm:big-demands}
Let $\eps>0$. Assuming that all terminals are \emph{big} w.r.t.\ $\eps$, there is a polynomial time $(1+\eps)$-approximation algorithm for the unsplittable Euclidean capacitated vehicle routing problem.
\end{theorem}
Notice that the special case with big terminals only remains NP-hard by a reduction from 3-Partition.\footnote{Recall that in the 3-Partition problem, 
we are given as input positive integers $a_1,\dots,a_{3n},b$, such that $b/4<a_i<b/2$ for all $i$, and such that $\sum_{i=1}^{3n}a_i=nb$.
We must decide whether there exists a partition of $\{1,\dots,3n\}$ into $n$ sets $T_j$ such that $\sum_{i\in T_j} a_i=b$ for all $j=1,\dots,n$.
The 3-Partition problem is NP-complete~\cite{williamson2011design}.}
See Table \ref{tab:unsplittable} for a comparison of our results with the closely related previous work. 

\begin{table}[H]
\centering
\renewcommand{\arraystretch}{1.4}
\begin{tabular}{ |l|p{4cm}|p{7cm}|} 
 \hline
  & equal demands & unequal demands\\\hline
arbitrary demands &           $2+\eps$~\cite{altinkemer1990heuristics} 
& 
2.694~\cite{friggstad2021improved} \newline 
{\bf $\bm{2+\eps}$ [this work: \cref{thm:main}]} 
\newline
$(3/2)$-hardness [folklore]
\\\hline
big demands &  $1+\eps$~\cite{haimovich1985bounds} & 2.694~\cite{friggstad2021improved}
\newline 
{\bf $\bm{1+\eps}$ [this work: \cref{thm:big-demands}}]
\newline
NP-hard [folklore]
\\\hline
\end{tabular}
\caption{
Approximation ratios for the Euclidean CVRP. 
}

\label{tab:unsplittable}
\end{table}

\subsection{Related Work}
\label{sec:related}
\paragraph{Equal Demands.}
The special case when all the demands are equal is called the \emph{unit demand} version of the Euclidean CVRP. 
Let the (equal) demand be $1/k$ where $k$ is a positive integer.
For general $k$, the best polynomial-time approximation factor is  $2+\eps$~\cite{altinkemer1990heuristics}.
For small values of $k$, PTASs are known. 
Haimovich and Rinnooy Kan~\cite{haimovich1985bounds} described a PTAS when $k$ is constant. 
Asano et al.~\cite{asano1997covering}, extending techniques in~\cite{haimovich1985bounds}, obtained a PTAS for $k=O(\log n/\log\log n)$. This was improved to $k\leq 2^{\log^{f(\eps)}(n)}$ for some function $f(\eps)$ depending on $\eps$ by
Adamaszek, Czumaj, and Lingas~\cite{adamaszek2010ptas}.
For higher dimensional Euclidean metrics, Khachay and Dubinin~\cite{khachay2016ptas} gave a PTAS for fixed dimension $\ell$ and $k=O(\log^{\frac{1}{\ell}}(n))$.

\paragraph{Relation with Heuristics in Practice.}
From a more practical perspective, we note that the standard setting in which the CVRP arises consists of real-life road networks. Such inputs are related to two theoretical models: Euclidean metrics (because distances between points are not exactly equal to but related to Euclidean distances) and planar graph metrics (because the road network is not exactly a planar graph but is  related to planar embedded graphs), thus those are the two theoretical settings which are most relevant in practice. This paper deals with the Euclidean CVRP; we leave the complementary planar graph CVRP for future work. 

Many algorithms use a 2-phase \emph{route-first-cluster-second} method, the ITP algorithm being the best-known one.
However, the reverse \emph{cluster-first-route-second} method is more common in practice, e.g., used by the famous \emph{OR-tools} developed by Google~\cite{ortools}.\footnote{See Lines 79--88 of the code: \url{https://github.com/google/or-tools/blob/stable/ortools/constraint_solver/routing.h}.}
As described in~\cite{cordeau2007vehicle},
\begin{quote}
\emph{In a \emph{cluster-first-route-second} method, customers are first grouped into clusters and the routes are then determined by suitably sequencing customers within each cluster. Different techniques have been proposed for the clustering phase, while the routing phase amounts to solving a TSP.}
\end{quote}
The \emph{cluster-first-route-second} method, although widely used in practice, does not have theoretical guarantees in terms of approximation factor except in special cases.
For example, Blauth, Traub, and Vygen~\cite{blauth2021improving} considered the worst-case instances for the ITP algorithm in metric spaces, and designed a \emph{cluster-first-route-second} algorithm with a theoretical guarantee for those worst-case instances.
In our work, we design another kind of \emph{cluster-first-route-second} algorithm that achieves a theoretical guarantee for general instances in Euclidean metrics. 
Indeed, existing heuristics of that method construct clusters of size at most equal to the tour capacity, so their routing creates an intra-cluster route for each cluster. Our main idea is to  construct finer clusters over the small terminals (of total demand roughly $\eps$). We also have an internal routing  to connect the terminals within each cluster into an intra-cluster route. In addition, we have an external routing phase so that each route is obtained by combining intra-cluster routes and big terminals. 

\section{Overview of our Techniques}\label{sec:overview}

Our approach uses a new way to cluster small terminals. In addition to that, we combine several ideas from the literature in a non-trivial way, including techniques developed in the framework of Euclidean TSP and Bin Packing.  

\paragraph{Big Terminals Only (\cref{sec:big}).}
Let us first consider the case of big terminals only (i.e., all the demands are at least some constant $\eps>0$). 

In the related Bin Packing problem, it is not hard to solve the problem exactly when (1) all items are big, i.e., their sizes are at least some constant $\eps>0$, and (2) there is a constant number of distinct item sizes. Exploiting the adaptive rounding technique, Karmarkar and Karp~\cite{karmarkar1982efficient} obtained a PTAS in the special case where all items are big (but the number of distinct sizes is arbitrary). 

Our PTAS for the case of big terminals (Theorem \ref{thm:big-demands}) is inspired by that approach. We start by observing that the problem can be solved exactly when (1) all demands are big, (2) there is a constant number of distinct demands, and in addition, (3) the terminals are placed at a constant number of distinct locations. Similarly to Bin Packing, up to losing a $1+\eps$ factor in the approximation, with adaptive rounding we get rid of the assumption (2). To get rid of the assumption (3), it is natural to partition the instance into subinstances of \emph{bounded distance} (\cref{def:bounded-dist,lem:bounded-distance}) and then discretize the Euclidean plane, similarly to the approach from~\cite{adamaszek2010ptas}. This leads to \cref{thm:big-demands}.

\paragraph{General Terminals (\cref{sec:general}).}
Now we consider the general case of big and small terminals together. 
To show \cref{thm:main}, we distinguish the following two cases.

\emph{Case 1: When the optimum solution $\OPT$ has sufficient many tours (\cref{sec:manytours}).} As before, we assume bounded distance w.l.o.g. The main novelty of our approach is to \emph{cluster} small terminals into groups with total demand roughly $\eps$ to reduce to the special case with big terminals only. The clustering process should put together small terminals that are near one another in the plane: to that end, we use a Voronoi cell decomposition of the plane, compute a traveling salesman tour on the small terminals in each cell, and then apply iterated tour partitioning with capacity $\epsilon$ to each such tour, with a virtual depot located at the center of the cell. 
For the analysis, thanks to the bounded distance and to the degree of discretization into cells,  we prove (\cref{lem:W}) that the cost of clustering small terminals is at most $(1+\eps)\cdot\opt$, where $\opt$ is the optimal cost.
Besides, thanks to the ``Assignment Lemma''  from~\cite{becker2019framework}, we prove (\cref{lem:coarse}) that an optimal solution covering the \emph{clustered} terminals and the big terminals has cost at most $(1+\eps)\cdot \opt$.
Since the clustered terminals behave the same as the big terminals, we apply our PTAS for big terminals only to obtain a solution of cost at most $(1+2\eps)\cdot \opt$ covering the clustered terminals and the big terminals.
Together with the cost of clustering small terminals, the overall cost of our solution to the initial instance is at most $(2+O(\eps))\cdot \opt$.


\emph{Case 2: When the optimum solution $\OPT$ has a bounded number of tours (Section \ref{sec:fewtours}).}
A $2+O(\eps)$ approximation is obtained by a simple adaptation of the techniques in \cite{arora1998polynomial,asano1997covering}. 

This completes the proof of \cref{thm:main}.

\paragraph{Open Questions.} To summarize, for the unsplittable Euclidean CVRP, in this paper we have designed algorithms that match the best known results for the special case in which all demands are equal. Can we go beyond that? 
Looking at \cref{tab:unsplittable}, we see that 
for the unsplittable Euclidean CVRP, \cref{thm:main} is not the end of the story: a gap remains in the approximation factor, between $2+\eps$ and $3/2$. A good starting point might be to focus on the special case when all demands are equal (first column of \cref{tab:unsplittable}). It is an interesting open question whether \cref{alg:main}, the algorithm in \cref{thm:main},  would yield a better-than-2 approximation for that case.

\section{Preliminaries}
\label{sec:notations}
Whenever needed, we will assume w.l.o.g.\ that $\eps>0$ is upper bounded by a sufficiently small constant. For a subset of terminals $U\subseteq V$, we let $\demand(U):=\sum_{v\in U}\demand(v)$.
For any point $p\in \mathbb{R}^2$, we let $\dist(p)$ denote the $p$-to-$O$ distance in the Euclidean plane.
\begin{definition}[bounded distance]
\label{def:bounded-dist}
Let $C:=(1/\eps)^{1/\eps}$.
We say that a set of terminals $V$ has \emph{bounded distance} if 
$\frac{D_{\max}}{D_{\min}}\leq C$, where $D_{\max}:=\max_{v\in V}\dist(v)$ and $D_{\min}:=\min_{v\in V}\dist(v)$.
\end{definition}

\begin{lemma}[adaptation from Theorem~3 of \cite{adamaszek2010ptas}]
\label{lem:bounded-distance}
There is a polynomial time algorithm to partition the set $V$ into disjoint subsets $V_1,\dots,V_m\subseteq V$, for some $m\in \mathbb{N}$, such that for each $i\in[1,m]$, $V_i$ has bounded distance, and for any $\rho\geq 1$, $\rho$-approximate solutions to the instances on $V_i$'s for all $i\in[1,m]$ yield a $(\rho+\eps)$-approximate solution to the instance on $V$. 
\end{lemma}


The following Assignment Lemma due to Becker and Paul~\cite{becker2019framework} will be used in \cref{lem:coarse} to combine subtours from different tours, while ensuring that the resulting tours violate the capacity only slightly.

\begin{lemma}[Assignment Lemma, Lemma~1 in \cite{becker2019framework}]
\label{lem:assignment}
Let $G=(A\uplus B,E)$ be a bipartite graph with $E\subseteq A\times B$, where for each $b\in B$, $N(b)\neq \emptyset$ denotes the set of neighbors of $b$. Each edge $(a,b)\in E$ has a weight $w(a,b)\geq 0$ and each vertex $b\in B$ has a weight $w(b)$ satisfying $0\leq w(b)\leq \sum_{a\in N(b)}w(a,b)$.
Then there exists a function $f: B\to A$ such that each vertex $b\in B$ is assigned by $f$ to a vertex $a\in N(b)$ and, for each vertex $a\in A$, we have
\[\sum_{b\in B\mid f(b)=a} w(b) - \sum_{b\in B\mid (a,b)\in E} w(a,b)\leq \max_{b\in B} \big\{w(b)\big\}.\]
\end{lemma}

The following lemma, which will be used to prove \cref{lem:coarse}, was initially due to Altinkemer and  Gavish~\cite{altinkemer1987heuristics} and rephrased by  Blauth, Traub, and Vygen~\cite{blauth2021improving}.

\begin{lemma}[\cite{altinkemer1987heuristics,blauth2021improving}]
\label{lem:rad}
Given an instance of the unsplittable capacitated vehicle routing problem with terminals $V$, depot $O$, and a traveling salesman tour $t_{\TSP}$ on $V\cup \{O\}$, there exists a feasible solution of cost at most $\cost(t_{\TSP})+\sum_{v\in V} 4\cdot\dist(v)\cdot \demand(v)$.
\end{lemma}

An elementary but critical step in our construction is to partition the plane into a collection of Voronoi cells as follows.
\begin{definition}
\label{def:center}
Let $Z$ denote the set of points in $\mathbb{R}^2$ whose polar coordinates $(r,\Theta)$ satisfy \[r=\left(1+\frac{\eps^2}{4}\cdot i\right)\cdot D_{\min}, \text{ where $i$ is an integer such that }0\leq i< K_1:=\frac{4C}{\eps^2}\] and 
\[\Theta= \frac{2\pi}{K_2}\cdot j, \text{ where $j$ is an integer such that }0\leq j< K_2:=\frac{8\pi C}{\eps^2}.\]
The points in $Z$ are called \emph{centers}. For each $z\in Z$, the \emph{cell} $c(z)$ of $z$ is the set of points $p\in \mathbb{R}^2$ such that $\dist(p)\in[D_{\min},D_{\max}]$ and $p$ is closer to $z$ than to any other center (breaking ties arbitrarily). 
The \emph{boundary} of $c(z)$ is denoted by $\bd(z)$. 
\end{definition}
\begin{fact}\label{fact:centers}
The number of centers in $Z$ is at most $K:=\frac{32\pi\cdot C^2}{\eps^4}$.
For any center $z\in Z$ and any point $v\in c(z)$, the $z$-to-$v$ distance is at most $\frac{\eps^2\cdot D_{\min}}{2}$. Furthermore, $\bd(z)\leq \eps^2 \cdot D_{\min}$.
\end{fact}

\section{Algorithm for Big Terminals Only (Proof of \cref{thm:big-demands})}

\label{sec:big}
\cref{thm:big-demands} follows directly from  \cref{lem:bounded-distance} and the following \cref{lem:big-demands}. 

\begin{lemma}
\label{lem:big-demands}
Let $\eps>0$. \cref{alg:big-terminals} is a polynomial time $(1+\eps)$-approximation algorithm for the unsplittable Euclidean capacitated vehicle routing problem under the assumption that all terminals are big and the set of terminals has bounded distance.
\end{lemma}

\begin{algorithm}[h]
\caption{Algorithm for big terminals only. Let $I$ be the input instance.
Let $\beta=\eps^2/(4C)$.}
\label{alg:big-terminals}
\begin{algorithmic}[1]
\State
Move each big terminal in $I$ to the closest center
\For{each center $z\in Z$ that contains at least $1/\beta$ terminals}
    \State Let $Q(z)$ denote the set of terminals at $z$
    \State Sort the terminals in $Q(z)$ in non-decreasing order of their demands
    \State Partition $Q(z)$ into $1/\beta$ groups of equal cardinality\footnotemark
    \State Round the demand of each terminal $v\in Q(z)$ to the maximum demand in the group of $v$
\EndFor
\State Let $I'$ be the resulting instance of the unsplittable Euclidean CVRP
\State Compute an optimal solution $S'$ to $I'$ \Comment{\cref{lem:big-demands-2}}
\For{each tour $t'\in S'$ and for each  terminal $v$ covered by $t'$}
\State Let $z\in Z$ denote the center where $v$ is located
\State Add to $t'$ two copies of the edge  between $z$ and the corresponding terminal of $v$ in $I$   
\EndFor
\State \Return the resulting solution $S$
\end{algorithmic}
\end{algorithm} 
\footnotetext{When $|Q(z)|$ is not an integer multiple of $1/\beta$, each of the first groups has one less element than each of the remaining groups.}

In the rest of the section, we prove \cref{lem:big-demands}.

\begin{lemma}
\label{lem:big-demands-1}
Let $\opt(I')$ denote the cost of an optimal solution to the instance $I'$.
We have $\opt(I')\leq (1+\eps)\cdot \opt$. 
\end{lemma}

\begin{proof}
First, we analyze the cost increase due to moving terminals to the closest centers. 
Consider any tour $t$ in $\OPT$.
Since all terminals are big, $t$ contains at most $1/\eps$ terminals. 
By \cref{fact:centers}, the cost increase of moving each terminal to the closest center is at most $\eps^2\cdot D_{\min}$.
Thus the overall cost of moving the terminals on $t$ to the closest centers is at most $\eps\cdot D_{\min}\leq \eps\cdot \cost(t)/2$.
Summing over all tours in $\OPT$, the overall cost increase of moving terminals to the closest centers is at most $(\eps/2)\cdot \opt$.

Next, we analyze the extra cost due to the adaptive rounding. 
The analysis is similar to the one for adaptive rounding for bin packing, except that, for each center $z$, each terminal in the group with the largest demands at $z$ is connected to the depot by adding a separate tour.
Consider any center $z$.
If $|Q(z)|<1/\beta$, then there is no extra cost due to the adaptive rounding;
If $|Q(z)|\geq 1/\beta$, the number of terminals at $z$ in the group with the largest demands is at most 
$\lceil \beta\cdot |Q(z)|\rceil\leq \beta\cdot|Q(z)|+1\leq 2\beta\cdot|Q(z)|$.
So the extra cost due to the adaptive rounding is at most $2\beta\cdot|Q(z)|\cdot 2 D_{\max}$.
Summing over all centers $z$, we have the overall extra cost is at most \[2\beta\cdot n \cdot 2 D_{\max}=
 (\eps^2/(2C)) \cdot n\cdot 2 D_{\max}\leq 
(\eps^2/2) \cdot n\cdot 2 D_{\min}\leq (\eps/2)\cdot \opt,\] using the fact that $\opt\geq \eps \cdot n\cdot 2 D_{\min}$ (since $\OPT$ has at least $\eps \cdot n$ tours of length at least $2D_{\min}$ each). 

The claim follows by summing the extra costs in both phases.
\end{proof}

\begin{lemma}
\label{lem:big-demands-2}
The optimal solution $S'$ for $I'$ can be computed in polynomial time.
\end{lemma}

\begin{proof}
There are only $O_\eps(1)$ centers, so there are only $O_\eps(1)$ positions for the terminals in $I'$.
Thanks to the adaptive rounding of the demands, each center has $O_\eps(1)$ distinct demands in $I'$.  Hence $I'$ can be described concisely by $O_\eps(1)$ pairs of ``position'' and ``demand'' with an integer for each pair giving the number of terminals located at that position and with that demand.
Moreover, since the demands are big, at most $1/\eps$ terminals are covered by each tour, hence there are only $O_\eps(1)$ possible distinct types of tours. 
Thus the instance $I'$ can be solved optimally in polynomial time by exhaustive search.  
\end{proof}

\begin{lemma}
\label{lem:big-demands-conversion}
$\cost(S)\leq (1+\eps/2)\cdot \cost(S')$. 
\end{lemma}
\begin{proof}
Consider a given tour $t'\in S'$, and let $t$ be the corresponding tour in $S$. Tour $t'$ covers at most $1/\eps$ terminals, and the increase of the length of $t$ w.r.t. $t'$ due to each such terminal $v$ is at most $\eps^2\cdot D_{\min}\leq \eps^2 \cdot \cost(t')/2$ by Fact \ref{fact:centers}. The claim follows. 
\end{proof}

\cref{lem:big-demands} follows from \cref{lem:big-demands-1,lem:big-demands-2,lem:big-demands-conversion}.

\section{Algorithm for General Terminals (Proof of \cref{thm:main})} 

\label{sec:general}
\cref{thm:main} follows directly from \cref{lem:bounded-distance} and the following  \cref{thm:general-demands}. 
\begin{theorem}
\label{thm:general-demands}
Let $\eps>0$. \cref{alg:main} is a polynomial time $(2+O(\eps))$-approximation algorithm for the unsplittable Euclidean capacitated vehicle routing problem assuming that the set of terminals has bounded distance.
\end{theorem}

\begin{algorithm}[h]
\caption{Algorithm for general terminals. Here $\thr=K\cdot C/\eps=32\pi C^3/\eps^5$.}
\label{alg:main}
\begin{algorithmic}[1]
\State $Y\gets \sum_{v\in V} \demand(v)$
\If{$Y\geq \thr$}
    \State $S\gets$ solution computed by \cref{alg:manytours}
\Else
    \State $S\gets$ solution computed by \cref{alg:fewtours}
\EndIf
\Return $S$
\end{algorithmic}
\end{algorithm}

In the rest of this section we prove \cref{thm:general-demands}. 

Define $\thr=K\cdot C/\eps=32\pi C^3/\eps^5$. We compute the overall demand $Y:=\sum_{v\in V}\demand(v)$.
We assume that $Y\geq 1$.\footnote{This assumption is without loss of generality, since otherwise the problem is equivalent to the Euclidean TSP and thus admits a PTAS~\cite{arora1998polynomial,mitchell1999guillotine}.}
It is easy to see that $Y$ is a 2-approximation on the number of tours in $\OPT$, i.e., the number of tours in $\OPT$ is in $[Y,2Y]$.\footnote{Indeed, if there are two tours in $\OPT$ whose total demand is at most 1, we can combine them into a single tour.}
If $Y\geq \thr$, then we apply Algorithm~\ref{alg:manytours} (see Section \ref{sec:manytours}), otherwise we apply Algorithm~\ref{alg:fewtours} (see Section \ref{sec:fewtours}).

By \cref{thm:manytours} (see Section \ref{sec:manytours}), the solution computed by \cref{alg:manytours} has cost at most $(2+O(\eps))\cdot \opt$ when the number of tours in $\OPT$ is at least $\thr$.
By \cref{thm:fewtours} (see Section \ref{sec:fewtours}), the solution computed by \cref{alg:fewtours} has cost at most $(2+O(\eps))\cdot\opt$ when the number of tours in $\OPT$ is at most $2\cdot\thr$.
Thus the cost of the solution $S$ is at most $(2+O(\eps))\cdot \opt$.

\subsection{Case 1: $\OPT$ Has Sufficiently  Many Tours}\label{sec:manytours}
In this subsection, we prove the following theorem.
\begin{theorem}
\label{thm:manytours}
Let $\eps>0$. \cref{alg:manytours} is a polynomial time $(2+O(\eps))$-approximation algorithm for the unsplittable Euclidean capacitated vehicle routing problem under the assumptions that the set of terminals has bounded distance and the number of tours in $\OPT$ is at least $\thr$.
\end{theorem}

\begin{algorithm}[h]
\caption{Algorithm when $\OPT$ has sufficiently many tours.
Let $I$ be the input instance.}
\label{alg:manytours}
\begin{algorithmic}[1]
\For{each center $z\in Z$}
	\State Let $Q(z)$ denote the small terminals in the cell $c(z)$ of $z$
	\State Compute a $(1+\eps)$-approximate traveling salesman tour $t(z)$ on $Q(z)$ \Comment{\cite{arora1998polynomial,mitchell1999guillotine}}
	\State Partition $t(z)$ into subtours of consecutive terminals, each with total demand in $[\eps,2\eps]$, except possibly the last subtour with smaller demand; each subtour is called a \emph{segment} of $t(z)$
	\For{each segment of $t(z)$}
		\State Let $d$ be the total demand of the segment
		\If{$d<\eps$}
		    $d\gets \eps$
		\EndIf
		\State Create a \emph{clustered} terminal located at $z$ and of demand $d$
	\EndFor
\EndFor
\State Let $I'$ be the instance of the unsplittable Euclidean CVRP induced by big and clustered terminals
\State Compute an optimal solution $S'$ to $I'$ \Comment{\cref{lem:big-demands}}.
\For{each tour $t'\in S'$ and for each clustered terminal $v$ covered by $t'$}
\State Let $z\in Z$ denote the center where $v$ is located
\State Let $x$ denote the  segment of $t(z)$ corresponding to $v$
\State Add to $t'$ the segment $x$, plus the edges between the endpoints of $x$ and the center $z$
\EndFor
\State \Return the resulting solution $S$
\end{algorithmic}
\end{algorithm}

\paragraph{A New Clustering of Small Terminals.}
We consider a Voronoi-cell decomposition of the plane with respect to the centers $Z$, where $c(z)$ denotes the cell of $z$ (breaking ties arbitrarily). Let $Q(z)$ be the small terminals contained in $c(z)$. For each cell $c(z)$, we cluster the small terminals $Q(z)$ inside that cell such that each resulting cluster has total demand roughly $\eps$.
For each cluster of small terminals, we replace them by a single \emph{clustered} terminal at the center $z$ of that cell, whose demand is the total demand of those small terminals.
This results in a new instance $I'$ consisting uniquely of big terminals, which can be solved near-optimally by \cref{lem:big-demands}.

How do we achieve a good clustering of the small terminals in each cell $c(z)$? We first compute a traveling salesman tour on the small terminals in that cell, and then greedily partition that tour into segments of consecutive terminals.

\begin{fact}
\label{fact:manytours}
Assume that the number of tours in $\OPT$ is at least $\thr$.
We have $K\cdot 2D_{\max}\leq \eps\cdot \opt.$
\end{fact}
\begin{proof}
$K\cdot 2D_{\max}=\frac{\thr\cdot \eps}{C}\cdot 2D_{\max}\leq \thr\cdot\eps\cdot 2D_{\min}\leq \eps\cdot \opt$ since $\OPT$ uses at least $\thr$ tours of length at least $2D_{\min}$ each.
\end{proof}

The following lemma shows that the clustering of small terminals is not too expensive.

\begin{lemma}
\label{lem:W}
Let $W$ denote the total cost of the segments and their connections to the centers. Then  $W\leq (1+O(\eps))\opt$.
\end{lemma}

\begin{proof}
Consider a center $z\in Z$.
Let $t^*(z)$ denote an optimal traveling salesman tour on $Q(z)$.
Let $\opt(z)$ denote the cost of the part of $\OPT$ that is inside the cell $c(z)$ of $z$.
Recall that $\bd(z)$ is the boundary of $c(z)$.
By Karp~\cite{karp1977probabilistic} and using $\bd(z)\leq \eps^2\cdot D_{\min}$ (\cref{fact:centers}), we have \[\cost(t^*(z))\leq \opt(z)+\frac{3}{2}\cdot \bd(z)\leq  \opt(z)+\frac{3}{2}\cdot \eps^2\cdot D_{\min}.\]
Since the traveling salesman tour $t(z)$ computed in \cref{alg:manytours} is a $(1+\eps)$-approximation, we have 
\[\cost(t(z))\leq (1+\eps)\cdot \cost(t^*(z))\leq (1+\eps)\cdot \left(\opt(z)+\frac{3}{2}\cdot \eps^2\cdot D_{\min}\right).\]

Next, we analyze the costs of the connections to the center $z$.
Observe that for each segment of $t(z)$, the total cost to connect the two endpoints of the segment to $z$ is at most $\eps^2\cdot D_{\min}$ by \cref{fact:centers}.
Since the demand of each segment is at least $\eps$ excluding possibly one segment, the number of segments is at most $\frac{\demand(Q(z))}{\eps}+1$.
Thus the cost of the connections to the center $z$ is at most 
\[ \left(\frac{\demand(Q(z))}{\eps}+1\right)\cdot \eps^2\cdot D_{\min}.\]

Let $W(z)$ denote the total cost of the segments in $c(z)$ and their connections to $z$.
We have 
\[W(z)\leq (1+\eps)\cdot \left(\opt(z)+\frac{3}{2}\cdot \eps^2\cdot D_{\min}\right) + \left(\frac{\demand(Q(z))}{\eps}+1\right)\cdot \eps^2\cdot D_{\min}.\]
Summing over all centers $z\in Z$ and since the number of centers is at most $K$, we have 
\[W=\sum_z W(z)\leq (1+\eps)\cdot \opt  + \frac{\demand(V)}{\eps}\cdot \eps^2 \cdot D_{\min}+ 3 K \eps^2 \cdot D_{\min}. 
\]
Observe that $\opt\geq \demand(V)\cdot 2 D_{\min}$, since $\OPT$ must contain at least $\demand(V)$ tours of length at least $2D_{\min}$ each. Then \[\frac{\demand(V)}{\eps}\cdot \eps^2 \cdot D_{\min}\leq \frac{\eps\cdot\opt}{2}.\]
By \cref{fact:manytours},
\[
3K\eps^2\cdot  D_{\min}\leq 3K\cdot D_{\max}\leq \frac{3\eps\cdot\opt}{2}.\]
Therefore, $W\leq (1+3\eps)\cdot \opt$.
\end{proof}

\begin{lemma}
\label{lem:coarse}
Let $\opt(I')$ denote the cost of an optimal solution to the instance $I'$.
We have $\opt(I')\leq (1+O(\eps))\cdot \opt$. 
\end{lemma}

\begin{proof}
First, we construct a bi-criteria solution $\tilde S$ to the instance $I'$ where each tour has total demand at most $1+2\eps$. 
Then we transform $\tilde S$ into a solution $S'$ to the instance $I'$ without capacity violation, and we show that the cost of $S'$ is at most $(1+O(\eps))\cdot\opt$.

We say that a clustered terminal is \emph{good} if its corresponding segment has total demand at least $\eps$, and \emph{bad} otherwise\footnote{We remark that here we do not consider the rounding up to $\eps$ of the quantity $d$ in the algorithm.}. We construct the solution $\tilde S$ in two phases.
\begin{enumerate}
    \item 
Define a bipartite graph $G=(A\uplus B,E)$ as follows.
Let $A$ denote the set of tours in an optimal solution to the instance $I$.
Let $B$ denote the set of good clustered terminals.
There is an edge in $G$ between a tour $a\in A$ and a clustered terminal $b\in B$ if and only if $a$ covers a small terminal in the segment corresponding to $b$. We define $w(b)$ as the demand of the segment $b$. We also define $w(a,b)$ as the part of $w(b)$ which is covered by $a$. Notice that $w(b)\leq 2\eps$, $w(b)=\sum_{a\in N(b)}w(a,b)$, and $N(b)\neq \emptyset$ for each $b\in B$. We apply the Assignment Lemma (\cref{lem:assignment}) to obtain a function $f: B\to A$ such that each clustered terminal $b\in B$ is assigned to exactly one tour $a\in A$ with $(a,b)\in E$ and for each tour $a\in A$, the total demand of the big terminals on $a$ and the clustered terminals that are assigned to $a$ by $f$ is at most $1+2\eps$.

\item For each bad clustered terminal, we connect it to the depot by a separate tour.
\end{enumerate}

Now we analyze the cost of $\tilde S$.
Consider a tour $a\in A$ in the first phase.
Recall that each clustered terminal $b\in B$ is good, so has demand at least $\eps$.
Thus there are at most $1/\eps$ clustered terminals $b\in B$ with $f(b)=a$.
For each $b\in B$ such that $f(b)=a$, by definition, $a$ contains a small terminal $v$ in the segment corresponding to $b$.
The cost to connect $b$ to $a$ is at most twice the $v$-to-$b$ distance, which is at most $\eps^2\cdot D_{\min}$ by \cref{fact:centers}.
So the cost to connect $b$ to $a$ over all $b\in B$ such that $f(b)=a$ is at most $\eps\cdot D_{\min}$, which is at most $\eps\cdot\cost(a)$.
Summing over all tours $a\in A$, the overall cost of the connection is at most $\eps\cdot \opt$.
In the second phase, at most one separate tour is created for each cell. Since the number of cells is at most $K$, the total connection in this phase has cost at most $K\cdot 2 D_{\max}$, which is at most $\eps\cdot \opt$ by \cref{fact:manytours}.
Therefore, $\cost(\tilde S)\leq (1+2\eps)\cdot \opt$.

Next, we construct the solution $S'$ from the solution $\tilde S$.
Consider any tour $t$ in $\tilde S$.
If $t$ exceeds the capacity, then we remove from $t$ clustered terminals until $t$ is within the capacity.
Let $t'$ be the resulting tour.
We include $t'$ into $S'$.
Let $R$ be the set of removed clustered terminals from the tours in $\tilde S$.
We create additional tours to cover the terminals in $R$ as follows.
First, we connect those terminals to the depot by a traveling salesman tour $t_{\TSP}$ visiting $R$, such that the terminals in $R$ located at the same center are visited consecutively by  $t_{\TSP}$.
Next, we apply \cref{lem:rad} on $t_{\TSP}$ to obtain a set of tours within the tour capacity and covering the terminals in $R$.
We include those tours into $S'$.

From the construction and \cref{lem:rad}, 
\[\cost(S')\leq \cost(\tilde S)+\cost(t_{\TSP})+\sum_{v\in R}4\cdot\dist(v)\cdot \demand(v).\]
Since there are at most $K$ centers, all clustered terminals are located at centers, and  using \cref{fact:manytours}, 
\[\cost(t_{\TSP})\leq K\cdot 2D_{\max}\leq \eps\cdot \opt.\]
At the same time, we have 
\begin{align*}
    & \sum_{v\in R}4\cdot\dist(v)\cdot \demand(v) =\sum_{t\in \tilde S} \sum_{v\in R\text{ and } v\in t} 4\cdot\dist(v)\cdot \demand(v)\\
    \leq & \sum_{t\in \tilde S} \sum_{v\in R\text{ and } v\in t} 2\cdot\cost(t)\cdot \demand(v) 
    \leq  \sum_{t\in \tilde S} 4\eps\cdot 2\cdot\cost(t) =8\eps\cdot\cost(\tilde S).
\end{align*}
In the first inequality above we used the fact that $2\cdot\dist(v)\leq \cost(t)$. In the last inequality above we used the fact that $\demand(t)\leq 1+2\eps$ and $\demand(t')>1-2\eps$, where $t'$ is the tour obtained from $t$ by removing a minimal subset of clustered terminals (covered by $t$) in order to enforce $\demand(t')\leq 1$. Thus \[\sum_{v\in R\text{ and } v\in t}\demand(v)=\demand(t)-\demand(t')\leq 4\eps.\] 
Altogether, for $\eps>0$ small enough, 
\[\cost(S')\leq \eps\cdot \opt+ (1+8\eps)\cost(\tilde S)<(1+12\eps)\cdot \opt.\]
\end{proof}

The cost of output solution in Algorithm~\ref{alg:manytours} is at most $W+(1+\eps)\cdot \opt(I')$ by Lemma \ref{lem:big-demands}. This is at most $(2+O(\eps))\cdot \opt$ by \cref{lem:W,lem:coarse}.
Hence Algorithm~\ref{alg:manytours} has approximation ratio $2+O(\eps)$.

\subsection{Case 2: $\OPT$ Has a Bounded Number of Tours}\label{sec:fewtours}
In this subsection, we prove the following theorem.

\begin{theorem}
\label{thm:fewtours}
Let $\eps>0$.
\cref{alg:fewtours} is a polynomial time $(2+O(\eps))$-approximation algorithm for the unsplittable Euclidean capacitated vehicle routing problem under the assumption that the number of tours in $\OPT$ is at most $2\cdot \thr$.
\end{theorem}

\begin{algorithm}[h]
\caption{Algorithm when $\OPT$ has a bounded number of tours.}
\label{alg:fewtours}
\begin{algorithmic}[1]
\State Round each demand down to the next smaller integer multiple of $\frac{1}{2n}$, obtaining an instance $\hat I$.
\State Compute a $(1+\eps)$-approximate solution $\hat S$ to $\hat I$\Comment{\cref{lem:Arora-modified}}
\State $S\gets \emptyset$
\For{each tour $t$ in $\hat S$}
\State Make two copies $t_1$ and $t_2$ of $t$
\State Let $v_1,\dots,v_\ell$ be the terminals covered by $t$ in non-increasing order of their unrounded demands
\State Let $i$ be the largest integer in $[1,\ell]$ such that $\sum_{j=1}^{i}\demand(v_j)\leq 1$
\State Let $t_1$ cover the demands of $v_1,\dots,v_i$
\State Let $t_2$ cover the demands of $v_{i+1},\dots,v_\ell$
\State $S\gets S\cup\{t_1,t_2\}$
\EndFor
\State \Return $S$
\end{algorithmic}
\end{algorithm}

\begin{lemma}
\label{lem:Arora-modified}
Let $\eps>0$. There is a polynomial time $(1+\eps)$-approximation algorithm for the unsplittable Euclidean capacitated vehicle routing problem under the assumptions that the number of tours in $\OPT$ is at most $2\cdot \thr$ and each demand is an integer multiple of $\frac{1}{2n}$.
\end{lemma}

\begin{proof}[Proof sketch of \cref{lem:Arora-modified}]
The proof is a straightforward adaptation of the algorithm of Asano et al.~\cite{asano1997covering} for the unit-demand case (i.e., for the case when all demands are equal) of the Euclidean CVRP, which in turn is a simple adaptation of Arora's algorithm~\cite{arora1998polynomial} for the Euclidean TSP.

Recall that Arora's algorithm defines a randomized hierarchical quadtree decomposition, such that a near-optimal solution intersects the boundary of each square only $O_\eps(1)$ times and those crossings happen at one of a small set of prespecified points, called \emph{portals}, and then uses a polynomial time dynamic program to find the best solution with this structure.

In \cite{asano1997covering} it was observed that, when the number of tours in $\OPT$ is $O_\eps(1)$, there is a near-optimal solution in which the overall number of subtours passing through each square (via portals) is $O_\eps(1)$. Furthermore, for all equal demands, one can guess the number of terminals covered by each such subtour within a polynomial number of options. This leads to a polynomial number of configurations of subtours inside each square, which ensures the polynomial running time of a natural dynamic program.

Similarly to \cite{asano1997covering}, when the demands are integer multiples of $\frac{1}{2n}$, one can guess the total demand covered by each subtour in the above sense. The details are left to the reader.
\end{proof}

\begin{proof}[Proof of \cref{thm:fewtours}]
From \cref{lem:Arora-modified}, we obtain in polynomial time a set of tours $\hat S$ covering the rounded demands of all terminals such that the total cost of the tours in $\hat S$ is at most $(1+\eps)\cdot \opt$.

First, we show that the solution $S$ constructed from $\hat S$ is feasible for the unsplittable Euclidean CVRP. 
Consider two tours $t_1$ and $t_2$ in $S$ corresponding to a tour $t\in \hat S$.
Let $v_1,\dots,v_\ell$ be the terminals covered by $t$ in non-increasing order of their unrounded demands.
From the construction, $t_1$ is within the tour capacity. It suffices to show that $t_2$ is within the tour capacity.
If $\sum_{j=1}^{\ell}\demand(v_j)\leq 1$, then $t_2$ covers no demand and is trivially within the tour capacity.
Next consider the case when $\sum_{j=1}^{\ell}\demand(v_j)> 1$.
Since the rounding creates an extra demand of at most $\frac{1}{2n}$ at each terminal and the number of terminals $\ell$ is at most $n$, we have 
$\sum_{j=1}^{\ell}\demand(v_j)\leq 1+\frac{\ell}{2n}\leq \frac{3}{2}$.
To show that $t_2$ is within the tour capacity, it suffices to show that the overall demand on $t_1$ is at least $\frac{1}{2}$.
There are two cases. 
\begin{description}
\item[Case 1: $\demand(v_1)\geq \frac{1}{2}$.] Then $\demand(t_1)\geq \demand(v_1)\geq \frac{1}{2}$.
\item[Case 2: $\demand(v_1)<\frac{1}{2}$.]
By the choice of $i$, we have $\sum_{j=1}^{i+1}\demand(v_j)>1$.
Since $\demand(v_{i+1})\leq \demand(v_1)<\frac{1}{2}$, we have $\demand(t_1)=\sum_{j=1}^{i}\demand(v_j)>\frac{1}{2}$.
\end{description}
So both $t_1$ and $t_2$ are within the tour capacity.
Thus the solution $S$ is feasible.  

For each tour $t\in \hat S$, each of the two corresponding tours $t_1$ and $t_2$ has cost at most the cost of $t$.
Thus $\cost(S)\leq 2\cdot \cost(\hat S)\leq 2(1+\eps)\cdot\opt$.

Hence Algorithm~\ref{alg:fewtours} returns a feasible solution to the unsplittable Euclidean CVRP with cost at most $2(1+\eps)\cdot \opt$. 
\end{proof}

\subsubsection*{Acknowledgement}
We thank Vincent Cohen-Addad for helpful preliminary discussions.

\bibliographystyle{alpha}

\bibliography{references}

\end{document}